\documentclass[12pt]{article}
\usepackage[latin1]{inputenc}
\usepackage{graphicx}
\usepackage{epsfig}
\usepackage{color}
\usepackage{lscape}
\usepackage{verbatim}
\usepackage{amsthm, amssymb}
\usepackage{amsmath}
\newtheorem{Lem}{Lemma}
\newtheorem{theorem}{Theorem}
\newtheorem{Cor}{Corollary}

\def\polylog{\operatorname{polylog}}

\title{Faster Shortest Path Algorithm for $H$-Minor Free Graphs with Negative Edge Weights}
\author{Christian Wulff-Nilsen
        \footnote{School of Computer Science,
                  Carleton University,
                  \texttt{koolooz@diku.dk},
                  \texttt{http://cg.scs.carleton.ca/$_{\widetilde{~}}$cwn/}.
                  Research partially supported by NSERC and MRI.}}

\date{}
\begin{document}

\maketitle
\begin{abstract}
Let $H$ be a fixed graph and let $G$ be an $H$-minor free $n$-vertex graph with integer edge weights
and no negative weight cycles reachable from a given vertex $s$. We present an algorithm that computes
a shortest path tree in $G$ rooted at $s$ in $\tilde{O}(n^{4/3}\log L)$ time, where $L$ is the absolute value
of the smallest edge weight. The previous best bound was $\tilde{O}(n^{\sqrt{11.5}-2}\log L) = O(n^{1.392}\log L)$.
Our running time matches an earlier bound for planar graphs by Henzinger et al.
\end{abstract}

\section{Introduction}
Computing shortest paths in graphs is a fundamental algorithmic problem. Two classical single source shortest path (SSSP) algorithms
are Dijkstra's algorithm and the algorithm of Bellman-Ford. For a graph with $n$ vertices and $m$ edges and no negative weight
cycles reachable from the source, the algorithm of Bellman-Ford finds a shortest path tree in $O(mn)$ time. Dijkstra's algorithm is
faster with a running time of $O(m + n\log n)$ but it only works if all edge weights are non-negative. For graphs with integer
edge weights, Goldberg's algorithm~\cite{Goldberg} solves the SSSP problem in $O(m\sqrt n\log L)$ time, where $L$
is the absolute value of the smallest edge weight.

Faster algorithms are known for special classes of graphs. For planar graphs with non-negative edge weights, Henzinger
et al.~\cite{SSSPPlanar} showed that the SSSP problem can be solved in linear time. For planar graphs with
integer edge weights (negative and non-negative), they gave an $\tilde{O}(n^{4/3}\log L)$ time algorithm. Faster
SSSP algorithms for planar graphs with negative edge weights (not necessarily integers) have since been found. Currently,
$O(n\log^2/\log\log n)$ is the best known time bound~\cite{CWN}.

The most studied graph class in modern graph theory is the class of \emph{$H$-minor free graphs}. A
graph $G'$ is a \emph{minor} of a directed or undirected graph $G$ if $G'$ can be obtained from a subgraph of the
undirected version of $G$ by contracting edges. For some fixed graph $H$, the class of $H$-minor free graphs is the
class of graphs that do not contain $H$ as a minor.

All planar graphs belong to this class as they are $K_5$-minor free (and $K_{3,3}$-minor free). All $H$-minor free graphs are
sparse, i.e., $m = O(n)$~\cite{SparseHMinor}. An $O(n)$ time SSSP algorithm was given in~\cite{LinTimeSSSPHMinor} for $H$-minor
free graphs with non-negative edge weights. When negative edge weights are allowed and when all edge weights are integers,
$O(n^{3/2}\log L)$ time is achievable using Goldberg's algorithm. This bound was recently improved to
$\tilde{O}(n^{\sqrt{11.5}-2}\log L) = O(n^{1.392}\log L)$ by Yuster~\cite{Yuster}.

Our contribution is an improvement of Yuster's time bound to $\tilde{O}(n^{4/3}\log L)$,
thereby matching the bound by Henzinger et al.\ for planar graphs. To obtain this speedup, we develop a faster algorithm to compute a
certain division of the input graph. This is plugged into Yuster's SSSP algorithm to get the improved time bound. Frederickson's
algorithm~\cite{APSPPlanar} and variants of it to find such a division for planar graphs have several applications and we believe
our algorithm may have similar applications for the class of $H$-minor free graphs.

The organization of the paper is as follows. In Section~\ref{sec:Prelim}, we give some basic definitions as well as some of
the tools that we need to obtain our result. We show how to efficiently obtain a division of the input graph in
Section~\ref{sec:Div} and we use it to get our SSSP algorithm in Section~\ref{sec:ShortestPaths}. Finally, we make some
concluding remarks in Section~\ref{sec:ConclRem}.

\section{Preliminaries}\label{sec:Prelim}
We use definitions similar to those in~\cite{Yuster}. A \emph{separation} of a graph $G$ is a pair $(A,B)$ of vertex sets
$A,B\subseteq V(G)$ with $A\cup B = V(G)$ such that no edge has one endpoint in $A\setminus B$ and one endpoint in $B\setminus A$.
We call $A\cap B$ a
\emph{separator} of $G$. Assign a non-negative weight $w(v)$ to each vertex $v\in V(G)$. For each $U\subseteq V(G)$, define
$w(U) = \sum_{v\in U} w(v)$. If $n$ is the number of vertices of $G$, we say that $G$ has an \emph{$(f(n),\alpha)$-separator} if there
is a separation $(A,B)$ with $|A\cap B|\leq f(n)$ and $w(A\setminus B),w(B\setminus A)\leq \alpha w(V(G))$. A family of graphs
closed under subgraphs satisfies an \emph{$(f(n),\alpha)$-separator theorem} if every vertex-weighted $n$-vertex graph in the
family has an \emph{$(f(n),\alpha)$-separator}. We say that the family has an \emph{$(f(n),\alpha, T(n))$-separator algorithm}
if it satisfies an $(f(n),\alpha)$-separator theorem and an $(f(n),\alpha)$-separator can be constructed in $T(n)$ time.

A classical theorem by Lipton and Tarjan~\cite{SeparatorPlanar} states that the family of planar graphs has an
$(O(\sqrt n),2/3,O(n))$-separator algorithm. Alon et al.~\cite{SeparatorHMinor} showed that the family of $H$-minor free graphs
has an $(O(\sqrt n),2/3,O(n^{3/2}))$-separator algorithm. This was generalized by Reed and Wood~\cite{FastSeparatorHMinor}. We state
their result in the following lemma as
we will use it extensively in our paper. It gives a trade-off between the size of the separator and the time to find it.
\begin{Lem}\label{Lem:SepSizeTimeTradeoff}
Let $\gamma\in[0,1/2]$ be fixed and let $H$ be a fixed graph. The family of $H$-minor free graphs has an
$(O(n^{(2 - \gamma)/3}),2/3,O(n^{1 + \gamma}))$-separator algorithm.
\end{Lem}

For a graph $G$, a \emph{region} (of $G$) is a subset of vertices of $G$ induced by a subset of edges of $G$. Partitioning
$E(G)$ into $k$ sets induces a set of $k$ (possibly overlapping) regions. Given such a set of regions, a \emph{boundary vertex}
is a vertex belonging to more than one region. Let $n$ be the number of vertices of $G$. An \emph{$(r,p)$-division} of $G$ is
a partition of $E(G)$ into $O(n/r)$ subsets such that each of the induced regions contains at most $r$ vertices and $O(p)$ boundary
vertices.

For $\gamma > 0$, Yuster~\cite{Yuster} applied Lemma~\ref{Lem:SepSizeTimeTradeoff} recursively to obtain an
$(r,r^{(2 - \gamma)/3})$-division of an
$n$-vertex $H$-minor free graph in $O(n^{1 + \gamma})$ time using ideas of Frederickson~\cite{APSPPlanar}. We will show how the
same lemma can be applied to obtain such a division in $O(nr^{\gamma}\log (n/r))$ time for all $r$ and in $O(nr^{\gamma})$ time for
$r = n^{\Omega(1)}$. This is no worse than Yuster's bound and asymptotically better when $r = o(n)$. In
Section~\ref{sec:ShortestPaths}, we apply this result to get the improvement over Yuster's shortest path algorithm.

The intuition behind our algorithm is as follows. Subgraphs at the top-levels of the
recursion are large so finding good separators for them is expensive. For these subgraphs, we therefore pick a value $\gamma'$ smaller
than $\gamma$ when applying Lemma~\ref{Lem:SepSizeTimeTradeoff}. As we move down the recursion tree, we increase $\gamma'$
such that it slowly approaches $\gamma$, thereby finding increasingly better separators. We can afford this since subgraphs
are smaller at the lower recursion levels and we prove that the chosen $\gamma'$-values suffice to
give the desired $(r,r^{(2 - \gamma)/3})$-division. Our approach differs from Yuster's which
applies Lemma~\ref{Lem:SepSizeTimeTradeoff} with fixed $\gamma' = \gamma$ through the recursion.

\section{Computing an $(r,p)$-division}\label{sec:Div}
In the following, let $p = r^{(2 - \gamma)/3}$. We will show how to find an $(r,p)$-division of an $H$-minor free
graph $G$ with $n$ vertices in $O(nr^{\gamma}\log(n/r))$ time.

We shall first divide $G$ into $\Theta(n/r)$ regions
each containing at most $r$ vertices such that the sum of boundary vertices over all regions is $O(pn/r)$ (so a boundary
vertex belonging to $k$ regions contributes with the value $k$ to the sum). We shall refer
to such a division as a \emph{weak $(r,p)$-division}. Observe that any $(r,p)$-division is a weak $(r,p)$-division.
The converse is not true as some regions in a weak $(r,p)$-division may not have an $O(p)$ bound on the number of boundary vertices.
We will later show how to efficiently convert a weak $(r,p)$-division into an $(r,p)$-division.

\subsection{Computing a weak $(r,p)$-division}\label{subsec:WeakDiv}
We now describe how to compute a weak $(r,p)$-division. As mentioned earlier, we will pick different $\gamma'$-values for
the various subgraphs obtained in the recursive subdivision. More precisely, for a subgraph of size $N > r$, we apply
Lemma~\ref{Lem:SepSizeTimeTradeoff} with the value $\gamma(N,r,p)$ satisfying
\begin{eqnarray}\label{eqn:Gamma}
  p(N/r)^{2/3} = N^{(2 - \gamma(N,r,p))/3}.
\end{eqnarray}
Clearly, $\gamma(N,r,p)$ can always be chosen so that it satisfies~(\ref{eqn:Gamma}). We will later
show that $\gamma(N,r,p)\in[0,1/2]$ for all $N > r$ so that Lemma~\ref{Lem:SepSizeTimeTradeoff} can be applied.

These choices of $\gamma(N,r,p)$-values result in a division of $G$ into regions each of size at most $r$. To bound the
sum of boundary vertices over all regions in this division, we use ideas of Frederickson (see proof of Lemma 1 in~\cite{APSPPlanar}).
For each boundary vertex $v$, let $b(v)$ be one less than the number of regions containing $v$. Let $B(n,r,p)$ be the sum of
$b(v)$ over all boundary vertices $v$. We will show that for $n > \frac 1 3 r$,
\begin{equation}
  B(n,r,p) \leq cp\frac n r - dp\left(\frac n r\right)^{2/3}\label{eqn:B}
\end{equation}
for constants $c$ and $d$. Note that since we do not partition a region of size $\leq r$, all regions have size greater
than $\frac 1 3 r$ so we only consider values $n > \frac 1 3 r$.

If we can show~(\ref{eqn:B}), this will give the desired bound on the
sum of boundary vertices over all regions and it will also imply that
there are $\Theta(n/r)$ regions in total since the sum of the number of vertices over all regions is
$n + B(n,r,p) = n + O(pn/r) = O(n)$ and each region is of size $\Theta(r)$.

To bound $B(n,r,p)$, we set up a recurrence relation. For $n \geq r$, we have, with $1/3\leq\alpha\leq 2/3$,
\begin{align*}
  B(n,r,p) & \leq c'p\left(\frac n r\right)^{2/3} + B\left(\alpha n + c'p\left(\frac n r\right)^{2/3},r,p\right)\\
           & \phantom{{} \leq} + B\left((1 - \alpha)n + c'p\left(\frac n r\right)^{2/3},r,p\right),
\end{align*}
where $c' > 0$ is a constant, and $B(n,r,p) = 0$ for $\frac 1 3 r < n < r$.

We prove~(\ref{eqn:B}) by induction on $n$. If we pick $c = \sqrt[3]3d$ 
then~(\ref{eqn:B}) holds for $\frac 1 3 r < n < r$. Now, assume that $n\geq r$ and that~(\ref{eqn:B}) holds for smaller values.
Then
\begin{align*}
B(n,r,p) & \leq c'p\left(\frac n r\right)^{2/3} + cp \frac n r + 2cp\frac{c'p(n/r)^{2/3}}r\\
         & \phantom{{}\leq} - \frac{dp}{r^{2/3}}\left[\left(\alpha n + c'p\left(\frac n r\right)^{2/3}\right)^{2/3} +
                              \left((1 - \alpha) n + c'p\left(\frac n r\right)^{2/3}\right)^{2/3}\right].
\end{align*}
We will show that the right-hand side in this inequality is at most the right-hand side in~(\ref{eqn:B}). This will follow from
\begin{align*}
c'p\left(\frac n r\right)^{2/3} + 2cp\frac{c'p(n/r)^{2/3}}r\leq
\frac{dp}{r^{2/3}}\left((\alpha n)^{2/3} + ((1 - \alpha)n)^{2/3} - n^{2/3}\right).
\end{align*}
Dividing by $p(n/r)^{2/3}$ on both sides gives
\[
  c' + \frac{2cc'}{r^{(1 + \gamma)/3}} = c' + 2cc'\frac p r \leq d\left(\alpha^{2/3} + (1 - \alpha)^{2/3} - 1\right).
\]
Since $\frac 1 3\leq\alpha\leq\frac 2 3$, we have $\alpha^{2/3} + (1 - \alpha)^{2/3} - 1\geq \epsilon$ with
constant $\epsilon = (1/3)^{2/3} + (2/3)^{2/3} - 1 > 0$. Also recall that $c = \sqrt[3]3d$. The induction step will thus follow
if we can show that
\[
  \frac{c'}d + \frac{2\sqrt[3]3c'}{r^{(1 + \gamma)/3}}\leq \epsilon.
\]
We may assume that both $r$ and $d$ are bounded from below by some large constant so we can make this equation hold. This completes
the proof by induction.

We have shown that the chosen $\gamma(N,r,p)$-values give the desired weak $(r,p)$-division. However, it only works if each such
value is in the interval $[0,1/2]$ since otherwise, Lemma~\ref{Lem:SepSizeTimeTradeoff} does not apply. Since $p\geq\sqrt r$ and
$r \leq N$,
\[
  p(N/r)^{2/3} \geq N^{2/3}/r^{1/6} \geq \sqrt N,
\]
which by equation~(\ref{eqn:Gamma}) implies $\gamma(N,r,p) \leq 1/2$.
To show that $\gamma(N,r,p) \geq 0$, we again apply equation~(\ref{eqn:Gamma}) and get
\[
  \log p + \frac 2 3 \log(N/r) = \frac 1 3 (2 - \gamma(N,r,p))\log N,
\]
implying that
\[
  \gamma(N,r,p) = 2 - 3\frac{\log p + \frac 2 3\log(N/r)}{\log N}.
\]
Hence,
\begin{eqnarray}\label{eqn:GammaPos}
  \gamma(N,r,p) \geq 0\Leftrightarrow \frac 2 3\log N \leq \frac 2 3\log r - \log p.
\end{eqnarray}
Since $p \leq r^{2/3}$, the equation on the right-hand side is satisfied. We note that
\begin{eqnarray}\label{eqn:EpsilonThird}
  \gamma(N,r,p) = \frac{2\log r - 3\log p}{\log N}.
\end{eqnarray}

\subsection{Running time}\label{subsec:Runtime}
Now, let us bound the running time to find a weak $(r,p)$-division. Let $R$ be a region that is separated in some recursive step of
the algorithm. By equation~(\ref{eqn:EpsilonThird}) and Lemma~\ref{Lem:SepSizeTimeTradeoff}, the time to find this separation is
\begin{align*}
  O(|R|^{1 + \gamma(|R|,r,p)}) & = O(|R|^{1 + (2\log r - 3\log p)/\log|R|})\\
                               & = O(|R|2^{2\log r - 3\log p})\\
                               & = O(|R|r^2/p^3)\\
                               & = O(|R|r^{\gamma}).
\end{align*}
Since a weak $(r,p)$-division has $O(pn/r)$ boundary vertices over all regions and since there are $O(\log(n/r))$ recursion levels,
the total size of regions generated in all recursion levels is $O((pn/r + n)\log(n/r)) = O(n\log(n/r))$. This gives the following lemma.
\begin{Lem}\label{Lem:WeakrDiv}
Let $\gamma\in[0,1/2]$ be fixed and let $H$ be a fixed graph. For $r\leq n$, a weak $(r,r^{(2-\gamma)/3})$-division
in an $H$-minor free graph with $n$ vertices can be computed in $O(nr^{\gamma}\log(n/r))$ time.
\end{Lem}

The running time in Lemma~\ref{Lem:WeakrDiv} can be improved slightly when $\gamma > 0$ and $r = n^{\Omega(1)}$ as we show in
the following.
We will redefine values $\gamma(N,r,p)$ so that
\begin{eqnarray}\label{eqn:Gamma2}
  p(N/r)^{2/3 + \epsilon} = N^{(2 - \gamma(N,r,p))/3},
\end{eqnarray}
where $\epsilon > 0$ is some small constant that we specify below.
We need to show that we get a weak $(r,p)$-division with these values instead of those in equation~(\ref{eqn:Gamma}).
Going through Section~\ref{subsec:WeakDiv}, we see that the $O(pn/r)$ bound on the total number of boundary vertices still
holds if $\epsilon$ is sufficiently small. We need to show that $0\leq \gamma(N,r,p) \leq 1/2$. Since
\[
  p(N/r)^{2/3 + \epsilon} \geq N^{2/3 + \epsilon}/r^{2/3 + \epsilon - 1/2} > \sqrt N,
\]
it follows from equation~(\ref{eqn:Gamma2}) that $\gamma(N,r,p) < 1/2$. From the same equation, we have
\[
  \gamma(N,r,p) \geq 0\Leftrightarrow \epsilon\log N \leq (2/3 + \epsilon)\log r - (2/3 - \gamma/3)\log r
                                                     = (\epsilon + \gamma/3)\log r.
\]
Since $r = n^{\Omega(1)}$, there is a constant $c > 0$ such that $r \geq n^c\geq N^c$ and we have
$(\epsilon + \gamma/3)c\log N\leq (\epsilon + \gamma/3)\log r$. Furthermore,
\[
  \epsilon\log N \leq (\epsilon + \gamma/3)c\log N
                   \Leftrightarrow c \geq \frac{\epsilon}{\epsilon + \gamma/3}.
\]
Since $\gamma > 0$, we can pick $\epsilon > 0$ sufficiently small to ensure that the rightmost inequality holds. Combining with the
above, we get $\gamma(N,r,p) \geq 0$, as desired.

Now, let us bound the running time of the weak $(r,p)$-division algorithm with the choices of $\gamma(N,r,p)$ in~(\ref{eqn:Gamma2}).
We may assume that $2/3 + \epsilon < 3/4$. Hence, when finding a separation of a size $N$ region, the size
of each of the two subregions is at most $\frac 2 3 N + c'p(N/r)^{2/3 + \epsilon}\leq \frac 2 3 N + c'N^{3/4}$. This is at most
$\frac 3 4 N$ if $N\geq (12c')^4$. Since we stop separating regions when they have size at most $r$ and since $r = n^{\Omega(1)}$,
we may assume that $N\geq (12c')^4$.

During the course of the algorithm, regions of various
sizes are generated in the recursion. For $i = 1,\ldots,\log_{4/3} n$, let $N_i = (4/3)^i$ and let $\mathcal P_i$ be the set of regions
of size between $N_i + 1$ and $N_{i+1}$.
For two regions $R,R'\in\mathcal P_i$, $\frac 3 4 |R| < |R'| < \frac 4 3 |R|$. By the above, neither region is obtained in any
recursive separation of the other so the only vertices they may share are boundary vertices. Since the total number of boundary vertices
generated is $O(n)$, the total size of all regions in $\mathcal P_i$ is $O(n)$.

Now, the time to find a separation of a region $R\in\mathcal P_i$ is
\begin{align*}
O(|R|^{1 + \gamma(|R|,r,p)}) & = O(|R|^{3 - 3\frac{\log p + (2/3 + \epsilon)\log(|R|/r)}{\log|R|}})\\
                              & = O(2^{3\log|R| - 3\log p - 3(2/3 + \epsilon)\log(|R|/r)})\\
                              & = O(r^{2 + 3\epsilon}|R|^{1 - 3\epsilon}/p^3)\\
                              & = O(r^{\gamma + 3\epsilon}|R|^{1 - 3\epsilon})\\
                              & = O(|R|r^{\gamma}(r/N_i)^{3\epsilon}).
\end{align*}
Over all regions of $\mathcal P_i$, this is $O(nr^{\gamma}(r/N_i)^{3\epsilon})$. Since we do not separate regions when $N_i\leq r$,
we only need to sum over those $i$ for which $N_i > r$ in order to get the total running time. This sum is a geometric series and we
get the following result.
\begin{Lem}\label{Lem:WeakrDiv2}
Let $\gamma\in(0,1/2]$ be fixed, let $H$ be a fixed graph, and let $r\leq n$. If $r = n^{\Omega(1)}$ then
a weak $(r,r^{(2-\gamma)/3})$-division in an $H$-minor free graph with $n$ vertices can be computed in $O(nr^{\gamma})$ time.
\end{Lem}

From a weak $(r,p)$-division, we find an $(r,p)$-division as follows. For each region containing more than $cp$
boundary vertices for some constant $c$, apply the separator algorithm of Lemma~\ref{Lem:SepSizeTimeTradeoff} with vertex
weights distributed evenly on the boundary vertices of the region. Repeat this process until every region has at most $cp$ boundary
vertices.
\begin{Lem}\label{Lem:WeakrDivTorDiv}
Let $\gamma\in[0,1/2]$ be fixed and let $H$ be a fixed graph. For $r\leq n$, the above procedure transforms a weak
$(r,r^{(2-\gamma)/3})$-division in an $H$-minor free graph with $n$ vertices into an $(r,r^{(2-\gamma)/3})$-division in
$O(nr^{\gamma})$ time.
\end{Lem}
\begin{proof}
The proof is very similar to that of Frederickson for planar graphs (see Lemma 2 and its proof in~\cite{APSPPlanar}).
In the weak $(r,p)$-division, let $t_i$ be the number of regions with exactly $i$ boundary vertices. With the notation in
Section~\ref{subsec:WeakDiv} and $V_B$ denoting the set of boundary vertices over all regions in the weak $(r,p)$-division, we have
\[
  \sum_i it_i = \sum_{v\in V_B}(b(v) + 1) \leq 2B(n) = O(pn/r).
\]

In the weak $(r,p)$-division, consider a region $R$ with $i > cp$ boundary vertices. When the above procedure finds a
separation $(R_1,R_2)$ of $R$, both $R_1$ and $R_2$ contain at most a constant fraction of the boundary vertices of $R$. Hence, after
$di/(cp)$ splits of $R$ for some constant $d$, all subregions will contain at most $cp$ boundary vertices. This will result
in at most $1 + di/(cp)$ subregions and at most $c'p$ new boundary vertices per split for some constant $c'$. We may assume
that $c'\leq c$. The total number of new boundary vertices introduced by the above procedure is thus
\[
  \sum_i (c'p)(di/(cp))t_i \leq d\sum_i it_i = O(pn/r)
\]
and the number of new regions is at most
\[
  \sum_i (di/(cp))t_i = (d/(c'p))\sum_i it_i = O(n/r).
\]
Hence, the procedure generates an $(r,p)$-division.
Each separation takes $O(r^{1 + \gamma})$ time. Since the number of separations is bounded by the number of new regions, we spend
a total of $O(nr^{\gamma})$ time.
\end{proof}

\begin{Cor}\label{Cor:rDiv}
Let $\gamma\in[0,1/2]$ be fixed and let $H$ be a fixed graph. For any $r\leq n$, an $(r,r^{(2-\gamma)/3})$-division
in an $H$-minor free graph with $n$ vertices can be computed in $O(nr^{\gamma}\log(n/r))$ time. If $\gamma > 0$ and
$r = n^{\Omega(1)}$, running time is $O(nr^{\gamma})$.
\end{Cor}
\begin{proof}
Follows immediately from Lemmas~\ref{Lem:WeakrDiv},~\ref{Lem:WeakrDiv2}, and~\ref{Lem:WeakrDivTorDiv}.
\end{proof}

\section{Shortest paths}\label{sec:ShortestPaths}
By applying Corollary~\ref{Cor:rDiv}, we get the following result.
\begin{theorem}
Let $H$ be a fixed graph and let $G$ be an $H$-minor free $n$-vertex graph with integer edge weights and
no negative weight cycles reachable from a given vertex $s$. Then a shortest path tree in $G$ rooted at
$s$ can be computed in $\tilde{O}(n^{4/3}\log L)$ time, where $L$ is the absolute value of the smallest
edge weight.
\end{theorem}
\begin{proof}
Yuster~\cite{Yuster} showed that a shortest path tree in $G$ rooted at $s$ can be computed in
\[
  \tilde{O}(\max\{T(n,\gamma), n^{\frac{13 - 2\gamma}{8 + 2\gamma}}\log L\})
\]
time, where $T(n,\gamma)$ is the time to compute an $(r,r^{(2-\gamma)/3})$-division for
$r = n^{3/(4 + \gamma)}$. With $T(n,\gamma) = O(n^{1 + \gamma})$ and $\gamma = \sqrt{11.5} - 3$,
Yuster obtained a time bound of $\tilde{O}(n^{\sqrt{11.5}-2}\log L)$.

We can apply Corollary~\ref{Cor:rDiv} for $\gamma > 0$ to get
$T(n,\gamma) = O(nr^{\gamma}) = O(n^{(4 + 4\gamma)/(4 + \gamma)})$ and we can compute a shortest path tree in
\[
  \tilde{O}(\max\{n^{(4 + 4\gamma)/(4 + \gamma)}, n^{\frac{13 - 2\gamma}{8 + 2\gamma}}\log L\})
\]
time. Picking $\gamma = 1/2$ gives the desired $\tilde{O}(n^{4/3}\log L)$ running time.
\end{proof}

Yuster mentions that his algorithm also gives a better time bound when shortest path trees are to be computed for
\emph{multiple} sources. However, it is well-known that once a shortest path tree has been found from one source,
all subsequent shortest path trees can be computed in the same graph but with a so called \emph{reduced weight function} which
ensures that all edge weights are non-negative (here, we assume w.l.o.g.\ that all vertices are reachable from the first
source). For details, see e.g.~\cite{SSSPPlanarKlein}.
Since a shortest path tree can be computed in linear time in this case~\cite{LinTimeSSSPHMinor}, the total time
to find shortest path trees from $k$ sources is $O(n^{4/3}\log L\polylog n + kn)$, which is faster than Yuster's approach.

\section{Concluding remarks}\label{sec:ConclRem}
For a fixed graph $H$, we gave an $\tilde{O}(n^{4/3}\log L)$ time algorithm for computing a shortest path tree in
an $n$-vertex $H$-minor free graph with integer edge weights where $L$ is the absolute value of the smallest edge weight.
This is an improvement of a previous bound of $\tilde{O}(n^{\sqrt{11.5}-2}\log L) = O(n^{1.392}\log L)$ by Yuster
and it matches an earlier time bound for planar graphs by Henzinger et al.

Our result is obtained from a faster algorithm to compute a certain division of an $H$-minor free graph. A similar type of division
has found numerous applications for planar graph problems. We believe our algorithm may find similar uses for
problems related to $H$-minor free graphs.

The fastest known shortest path algorithm in a planar $n$-vertex graph with arbitrary real edge weights has running time
$O(n\log^2n/\log\log n)$. Can we also get $\tilde{O}(n)$ running time for $H$-minor free graphs?

\end{document}